\newif\ifSC
\SCtrue
\ifSC
\documentclass[onecolumn,draftclsnofoot,12pt]{IEEEtran}
\else
\documentclass[10pt,twocolumn]{IEEEtran}
\fi
\usepackage{cite}
\usepackage{amsmath,amssymb,amsfonts}
\usepackage{algorithmic}
\usepackage{xcolor}
\usepackage{amssymb}
\usepackage{bbm}
\usepackage{bm}
\usepackage{amssymb,mathtools}
\usepackage{amsbsy}
\usepackage[font=small]{subcaption}
\usepackage[font=small]{caption}
\usepackage{graphicx} 

\def\BibTeX{{\rm B\kern-.05em{\sc i\kern-.025em b}\kern-.08em
		T\kern-.1667em\lower.7ex\hbox{E}\kern-.125emX}}
\newcounter{relctr} 
\everydisplay\expandafter{\the\everydisplay\setcounter{relctr}{0}} 

\newcommand\numeq[1]%
{\stackrel{\scriptscriptstyle(\mkern-1.5mu#1\mkern-1.5mu)}{=}}
\AtBeginDocument{}

\newcommand{\x}{\bm{x}}
\newcommand{\y}{\bm{y}}
\newcommand{\z}{\bm{z}}

\newcommand{\1}{\mathbbm{1}}
\newcommand{\pt}{\mathrm{p}}
\newcommand{\drm}{\mathrm{d}}

\newcommand{\ob}{\mathrm{o}}

\newcommand{\dv}{\mathrm{d}}

\newcommand{\B}{\mathcal{B}(\ob,r_\drm)}
\newcommand{\PGFL}{\mathcal{G}}

\newcommand{\rx}{\bm{X}}
\newcommand{\mo}{\overline{m}}

\newcommand{\A}{\mathcal{A}}

\newcommand{\num}{{N}}

\newcommand{\R}{\mathbb{R}}
\newcommand{\cF}{\overline{F}}
\newcommand{\matern}{Mat\'ern~}

\usepackage{bbm}
\usepackage{amssymb}
\usepackage{amsmath,amsthm}
\usepackage{color}
\usepackage{amsfonts}
\usepackage{graphicx}
\usepackage{verbatim}
\usepackage{epstopdf}
\usepackage{cite}
\usepackage{tabulary}
\usepackage{mathtools}
\usepackage{enumitem}

\newcommand{\set}[1]{\mathsf{#1}}
\newcommand{\dist}[1]{\|#1\|}

\newcommand{\palmprobx}[2]{\mathbb{P}^{#2}{\left[#1\right]}}

\newcommand{\dd}{\mathrm{d}}

\newcommand\expect[1]{\mathbb{E}\left[#1\right]}
\newcommand\prob[1]{\mathbb{P}\left[#1\right]}

\newcommand\indside[1]{\mathbbm{1}\left({#1}\right)}

\newcommand{\Ball}{\mathcal{B}}

\newcommand{\ie}{{\em i.e.}~}

\newcommand\summ{N}

\newcommand{\bigconditioned}{\left.\vphantom{\frac34}\right|}

\newcommand{\X}{\mathbf{X}}
\newcommand{\Y}{\mathbf{Y}}
\newcommand{\Z}{\mathbf{Z}}
\newcommand{\expS}[1]{\exp{\left(#1\right)}}

\def\home{\hbox{\kern3pt \vbox to13pt{}%
   \pdfliteral{q 0 0 m 0 5 l 5 10 l 10 5 l 10 0 l 7 0 l 7 5 l 3 5 l 3 0 l f
               1 j 1 J -2 5 m 5 12 l 12 5 l S Q }%
   \kern 13pt}}

\newtheorem{theorem}{Theorem}

\newtheorem{corollary}{Corollary}[theorem]
\newtheorem{remark}{Remark}
\begin{document}
\title{\LARGE $k$th Distance Distributions of $n$-Dimensional \matern Cluster Process}
\author{Kaushlendra Pandey, Abhishek K. Gupta\,\,\vspace{-2em}
	\thanks{K. Pandey and A. K. Gupta are with IIT Kanpur, India, 208016. Email:\{kpandey,gkrabhi\}@iitk.ac.in.}
} 

\maketitle

\begin{abstract}
In this letter, we derive the CDF (cumulative distribution function) of $k$th contact distance (CD) and nearest neighbor distance (NND) of the $n$-dimensional ($n$-D) \matern cluster process (MCP).
We present a new approach based on the probability generating function (PGF) of the random variable (RV) denoting the number of points in a ball of arbitrary radius
to derive its probability mass function (PMF). The proposed method is general and can be used for any point process with known probability generating functional (PGFL). We also validate our analysis via numerical simulations and provide  insights using the presented analysis. We also discuss two applications, namely- macro-diversity in cellular networks and caching in D2D networks, to study the impact of clustering on the performance. 
\end{abstract}

\IEEEpeerreviewmaketitle
\vspace{-1em}
\section{Introduction}
 Owing to the tractability and simplicity, Poisson point process (PPP) where each point is distributed independently of other points, has been a popular choice to model nodes' location in wireless networks.
However, in  modern networks, nodes may exhibit clustering or attraction toward each other. Examples include small cells  that are usually clustered
around hotspot areas such as downtown \cite{saha20183gpp} and device-nodes  located according to uneven population distribution in a  device-to-device (D2D) network\cite{afshang2016modeling}.  In such cases, the nodes can be more appropriately modeled as a clustered point process compared to PPP. Poisson cluster process is a widely used cluster process that has cluster centers (known as parent points) distributed as a PPP and each cluster is independent and identically distributed. Examples of PCP include MCP which is best suited for the modeling and analyses of network patterns in the urban area \cite{lee2013stochastic}. The interference, coverage, and rate analysis of various clustered wireless networks using MCP is presented in \cite{azimi2018stochastic,afshang2016modeling,wang2019performance,saha20183gpp}  based on the PGFL of MCP. 
To analyze various performance metrics of a wireless network, stochastic geometry utilizes the CDF of $k$th CD (\ie the distance of $k$th closest point from the origin) and $k$th NND (\ie the distance of $k$th closest neighbor from a typical point). The knowledge of these distributions is useful in many applications. Two important examples include the analysis of cache-enabled D2D networks, where the content of interest may be cached at the $k$th closest device {\cite{Afshang-2016} and the analysis of geolocation, where we may need ranging measurements from at least $k$ anchor nodes to obtain a position fix \cite{Schloemann-2016}. The distribution can be used to derive the Ripley's K function and the pair correlation function for any motion invariant process, which is useful in characterizing the point process (PP) of interferers in less tractable settings \cite{Mankar2019}. 
In a cellular network with $k$ successive interference cancellation \cite{zhang2014performance}, the interference at a typical user is due to base-stations (BSs) located further away than the $k$th closest BS and its coverage analysis requires the CDF of $k$th CD. In sensor networks, it can be used to compute the probability that an event is sensed by at least $k$ sensors \cite{PanGup2020}. 
Other applications include computation of  dominant interference,  optimum hops, power requirement to send data to the $k$th nearest device, analysis of cooperation diversity \cite{haenggi2005distances} and performance of localization of nodes. 
The CDF $F_{R_k}(r)$ of the $k$th CD  for any PP is equal to the probability that there are at least $k$ points inside the ball of radius $r$. Similarly, for the NND, CDF is equal to the probability that there are at least $k$ points inside the ball of radius $r$ under Palm \cite{AndGupDhi16}. 
For MCP, the CDFs of CD  and NND for $k=1$ is  reported in \cite{afshang2017nearest,pandey2019contact}. 
  To the best of our knowledge, the expression for the $k$th CD and NND of MCP for general $k$ has not been reported in the past  literature.  The main issue is in calculating the probability that there are  $m$ points inside a ball whose expression is not available in the literature.

In this letter, we present a new approach  to derive the CDF of $k$th CD and NND for a $n$-D MCP using PGF-PMF relation. Using  the PGFL of MCP, we first derive the expression for the PGF for a RV denoting the number of points in a ball of radius $r$  which is a non-negative positive integer valued RV. Utilizing the relation between PGF and PMF of the RV of this form, 
we derive  the PMF of this RV which tells the probability that there are $m$ points in this ball. Using the  PMF, we present the CDF 
of $k$th CD. The same exercise is repeated under conditioning on the occurrence of a point  to get  the CDF of $k$th NND. %
The technique presented in this paper is generic and can be used to derive the $k$th CD and NND of any PP with known PGFL and conditional PGFL including other variants of PCP. We considered MCP  to demonstrate this technique owing to its suitably to model modern wireless networks.  This PGF-PMF relation has been used in the past to compute the load distribution of Voronoi cells \cite{SinDhiJ2013}. We also discuss two applications, namely- macro-diversity in cellular networks and caching in D2D networks, and study the impact of clustering on the performance of these networks. 
\\
\textbf{Notation:} 
For a location $\x$,  $x=||\x||$.    $\Ball(\x,r)$ denotes a ball with radius $r$ and center $\x$. 
The notation $\A(r,r_\drm,x)$ denotes the volume of intersection of two $n$-D balls 
$\Ball(\ob,r)$ and $\Ball(\x, r_\drm)$. 
 $f^{(k)}(y)$ {denotes} the $k$th { derivative} of $f$ with respect to $y$. $f^{(0)}(y)=f(y)$. The volume of unit $n$-D ball  is $v_n={\pi^{n/2}}/{\Gamma(\frac{n}{2}+1)}$. Let $\beta(r)=2\min(r,r_\drm)$.

\section{\matern Cluster Process}
In this paper, we consider a  $n$-D MCP which is defined as follows. 
Let there be a parent PP which is a uniform PPP $\Phi_\pt=\{\X_i:\X_i\in\R^n\}$ with intensity $\lambda_{\pt}$. To each  point $\X_i$ in the parent PP, associate a daughter PP (also known as cluster) $\Phi_{\X_i}$ such that  $\Phi_{\X_i}$ is a finite homogeneous PPP with density $\lambda_\drm(\y)$ confined in a ball of $\B$ \ie 
\begin{align*}
\lambda_{\drm}{(\y)}&=\lambda_{\drm}\times\1(\dist{\y}\leq r_\drm).
\end{align*}
This means that the number of points in each daughter PP is  Poisson distributed  with mean $\mo=\lambda_\drm{v_n r_\drm^n}$.  
All daughter PPs are identically and independently distributed. 
For $i$-th daughter PP, let $ \Y^{i}_j$ denote the location of $j$th point. The corresponding point 
 in $\Phi$ is located at  $\Z_j^i=\X_i+\Y^{i}_j$. Note that $\Y_j^i$ is uniformly distributed in $\B$ independently of other points. 
The MCP $\Phi$ is defined as the following union
\begin{align*}
\Phi=\bigcup\nolimits_{\X_i\in\Phi_\pt}\X_i+\Phi_{\X_i}.
\end{align*}
MCP is characterized by its PGFL $\PGFL[v]$.  If $v:\R^{n}\rightarrow[0\,\,1]$ and $\int_{\R^{n}}(1-v(\x))\dv \x<\infty$, $\PGFL[v]$ is given as \cite{haenggibook}:
\begin{align}
\PGFL[v]&=\exp\left(-\lambda_{\pt}\int_{\R^{n}}\left(1-\mathcal{H}_{\x}[v] \right)\dv \x\right).\label{PGFL}\\
\text{with }\mathcal{H}_{\x}[v] &=\exp\left(-\lambda_\drm
\int_{\Ball(\ob,r_\drm)}\left(1-v(\x+\y)\right)\dv\y
\right)\nonumber
\end{align}
\section{$k$-th Contact Distance Distribution}\label{sec:CD:PMF}
Let $R_k$ denote the $k$th CD \ie distance of the $k$th closest point  of the MCP from an arbitrary point. As discussed, the  CCDF of $R_k$: $\overline{F}_{R_k}(r)$ is equal to the probability that there are at least $k$ points inside the ball $\Ball(\ob,r)$. Let $\summ$ be the RV denoting the number of point in $\Ball(\ob,r)$. Then, 
\begin{align}
&{F}_{R_k}(r)=1-\prob{\summ<k}=1-\sum\nolimits_{m=0}^{k-1}\mathbb{P}\left[\summ=m\right].\label{eq:cdf_rk_1}
\end{align}
To proceed further, we  need the probability that  $\Ball(\ob,r)$ contains $m$ points which is the PMF of $\summ$ at $m$. 
\subsection{PGF of $\summ$}
To get PMF of $\summ$, we will first present its PGF using PGFL of MCP as given in the  Theorem 1 (See Appendix \ref{Appn:1} for proof).
\begin{theorem}\label{Thm:1}
	The PGF $\mathcal{P}_{\summ}(s)$, of $\summ$ is  
	\begin{align}
	&\mathcal{P}_{\summ}(s)=\expect{s^{N}}=\exp(g(s))\label{PGF}, \ \ \ \ \ \ \ \text{where\,} \\
	&g(s)=v_n\lambda_{\pt}\left(\int_{x=0}^{r+r_\drm}ne^{\lambda_\drm\A(r,r_\drm,x)(s-1)}x^{n-1}\dv x-(r+r_\drm)^n\right)\nonumber.
	\end{align}
\end{theorem}
 The expression of $\A(r,r_\drm,x)$ for $n=1,2,3$ is well known \cite{pandey2019contact}.
For $n=1$,  the closed form expression for PGF of $\summ$ is given in  Corollary \ref{PGF_n_1}. For higher dimensions ($n>2)$, tight bounds for expressions can be computed using bounds over $\A(r,r_\drm,x)$. 
 Readers are advised to refer \cite{pandey2019contact} for the same.
{\begin{corollary}\label{PGF_n_1}
	For $n=1$, 
	\begin{small}
	\begin{align*}\A(r,r_\drm,x)=
	\begin{cases}
\beta(r) & \text{if }0\leq x\leq |r-r_\drm|\\
r+r_\drm-x &\text{if }|r-r_\drm|\leq x\leq (r+r_\drm)
	\end{cases}.
	\end{align*}\end{small}
Therefore, the PGF of $\summ$ is given by \eqref{PGF} with 
		\begin{align*}
	&g(s)=2\lambda_{\pt}\left[|r-r_\drm|e^{\lambda_{\drm}\beta(r)(s-1)}-(r+r_\drm)+\frac{e^{\lambda_{\drm}(s-1)\beta(r)}-1}{\lambda_\drm(s-1)}\right].
		\end{align*}		
\end{corollary}}

\subsection{PMF of $\summ$}
We now derive the PMF of $\summ$ from its PGF. 
We note that  $\summ$ is a non-negative integer valued RV. 
There exists the following relation between PGF and PMF 
of a non-negative integer valued RV, $N$ 
\begin{align}
\mathbb{P}(\summ=0)&=\mathcal{P}_{\summ}(s=0),
\end{align}
\begin{align}
\mathbb{P}\left(\summ=m\right)&=\left[{\mathcal{P}^{(m)}_{\summ}(s)}/{m!}\right]_{s=0}\label{PGF_PMF_relationship}.
\end{align}
$\mathcal{P}_{\summ}(s)$ is of the form $\expS{g(s)}$. Hence, to find the $m$-th derivative of $\mathcal{P}_{\summ}(s)$, we will use
the Faà di Bruno's formula which states that $m$-th derivative of $\exp(g(s))$ with respect to $s$ is given as
\begin{align}\label{Faa-di-bruno}
\!\!\!
&\frac{\dd^m }{\dd s^m }\expS{g(s)}=\exp\left(g(s)\right)\sum\frac{m!}{b_1!\cdots b_m!}
\left(\frac{g^{(1)}(s)}{1!}\right)^{\!\!b_1}\cdots\left(\frac{g^{(m)}(s)}{m!}\right)^{\!\!b_m}
\end{align}
where the sum is over set $\set{B}_m$ consisting of all $m-$tuples  $\{b_1\cdots b_m\}$ with $b_i\ge0$ and 
$b_1+2b_2...+mb_m=m$.
Let us denote $h_k(r)=g^{{(k)}}(0)/k!,\, \text{and}\,\, h_{0}(r)=g(0)$. 
Using the  the Faa di Bruno's formula \eqref{Faa-di-bruno} in the PGF of $\summ$  \eqref{PGF}, the PMF  of $\summ$ is given as
\begin{align}
&\prob{\summ=m}=e^{-\lambda_{\pt}v_n(r+r_\drm)^n}e^{h_0(r)}\sum_{\set{B}_m}\frac{\left({h_1(r)}\right)^{b_1}\cdots \left({{h_m(r)}}\right)^{b_m}}{b_1!\cdots b_m!}.\label{eq:pmf1a}
\end{align}

Now, it remains to compute $h_k(r)$. The $k$th derivative of $g(s)$- $g^{(k)}(s)$  can be computed as
\begin{align*}
	&g^{{(k)}}(s)=\lambda_{\pt}nv_n \int_0^{r+r_\drm}\!\!\!\! (\lambda_\drm \A(r,r_\drm,x))^{k}e^{\lambda_\drm \A(r,r_\drm,x)(s-1)} x^{n-1}\dv x.
\end{align*}
Therefore, 
\begin{align*}
&h_k(r)=
\frac{\lambda_{\pt}nv_n}{k!} \int_0^{r+r_\drm} (\lambda_\drm \A(r,r_\drm,x))^{k}e^{-\lambda_\drm \A(r,r_\drm,x)} x^{n-1}\dv x.
\end{align*}
Moreover,
\begin{align*}
&h_0(r)=
{\lambda_{\pt}nv_n} \int_0^{r+r_\drm} e^{-\lambda_\drm \A(r,r_\drm,x)} x^{n-1}\dv x.
\end{align*}

Using the values of $h_k(r)$ in \eqref{eq:pmf1a} and then using \eqref{eq:cdf_rk_1}, we get 
the CDF of $R_k$ which is given in the following Theorem.

\begin{theorem}\label{Thm:2}
	The CDF of $R_{k}$ is
	\begin{align}
	F_{R_{k}}(r)&=1-e^{h_0(r)-\lambda_{\pt}v_n(r+r_\drm)^n}\sum_{m=0}^{k-1}\left[
	\sum_{\set{B}_m}\frac{\left({h_1(r)}\right)^{b_1}\cdots \left({{h_m(r)}}\right)^{b_m}}{b_1!\cdots b_m!}
	\right].
	\end{align}	
\end{theorem}
\begin{remark}\label{rem:1}
	The set $\mathrm{B}_{m}=\{(b_1,b_2...b_m):b_1+2b_2...+mb_m=m \}$ of all possible $m-$tuples can be easily computed for any $m$. For example, for $m=0,1,2$, it is:
	\begin{align*}
\mathrm{B}_{0}=\{(0)\},\mathrm{B}_{2}=\{(1,0)\},\,\mathrm{B}_{2}=\{(2,0),(0,1)\}.
	\end{align*}
\end{remark}
Furthermore, in the analysis of wireless networks, the first CD (just the CD) $R_1$ ($k=1$) plays a crucial role. For example, if BSs location follows MCP, then for any user under maximum power association, $R_1$ denotes the distance of the serving BS 
and $R_2$ and $R_3$ denote the distance of  the two dominant interferers. Using Remark \ref{rem:1} and  Theorem \ref{Thm:2}, we can easily derive the CDF of $R_k$ for $k=1,2,3$ as given in the following Corollary.
\begin{corollary}\label{expressionfork=1,2,3}
	The CDF of $R_k$ for $k=1,2,3$ is:
	\begin{align*}
	&F_{R_1}(r)=1-e^{h_0(r)-\lambda_{\pt}v_n(r+r_\drm)^n},\\
	&F_{R_2}(r)=1-e^{h_0(r)-\lambda_{\pt}v_n(r+r_\drm)^n}(1+h_1(r)),\\
	&F_{R_3}(r)=F_{R_2}(r)-e^{h_0(r)-\lambda_{\pt}v_n(r+r_\drm)^n}\left({h_2(r)}+{h_1^2(r)}/2\right).
	\end{align*}
\end{corollary}
\section{$k$th Nearest Neighbor Distance Distribution}
The $k$th NND $R^{'}_{k}$ of a 
PP is defined as the distance between the typical point of the PP to its $k$-nearest neighbor. Without loss of generality, we assume that the typical point $\z_\ob$ of $\Phi$ is located at the origin $\ob$. The CDF of $R^{'}_k(r)$ is:
\begin{align}
&F_{R^{'}_{k}}(r)=1-\mathbb{P}[\Phi(\Ball(\ob,r))-1\leq k-1|\z_\ob=\ob\in\Phi]\nonumber\\
&=1-\mathbb{P}^{!\ob}[N\leq k-1]=1-\sum_{m=0}^{k-1} \palmprobx{\summ=m }{!\ob},\label{prob_cdf_nnd}
\end{align}
where $\mathbb{P}^{!\ob}[\cdot]$ denote the reduced Palm measure which is  the distribution of $\Phi\setminus\{\ob\}$ conditioned on occurrence of a point at $\ob$ \cite{AndGupDhi16,haenggibook}.
Similar to the CD, we will first compute the PGF of $\summ$ under reduced Palm and then derive PMF. 

\subsection{{The PGF of $N$ under reduced Palm}}
To derive PGF of $N$ under reduced Palm, we would need PGFL of MCP under reduced Palm which is defined as conditional PGFL of MCP excluding $\z_{\ob}$ conditioned on the occurance of a point  at $\z_{\ob}$,   and is given as \cite{ganti2009interference}, $\PGFL^{!}_{\z_\ob}[v]=$
\begin{align}
\mathbb{E}^{!{\z_\ob}}\left[\prod_{\Z_i\in\Phi}v(\Z_i)\right]
=\PGFL[v]\ \frac{1}{v_n r_\drm^{n}}\int_{\B}\!\!\!\!\! \mathcal{H}_{-\x_\ob}[v]\dv \x_\ob\label{eq:condpgfl}.
\end{align}
Here $\x_\ob$ denotes the cluster center of the typical cluster $\z_\ob$ belongs to.
Using the conditional PGFL of GPP, we derive the PGF of $\summ$ under reduced Palm which is given in Theorem \ref{thm:3}.  See Appendix \ref{appen_nnd_1} for the complete proof.

\begin{theorem}\label{thm:3}
	The PGF of the number of point ($N$) of $\Psi\setminus\{\ob\}$ conditioned on $\ob\in\Phi$, falling in the ball $\Ball(\ob,r)$ is 
	\begin{align*}
&\mathcal{P}^{!\ob}_{\num}(s)=\mathcal{P}_{\summ}(s)\int_{0}^{r_\drm}e^{(s-1)\lambda_\drm\A(r,r_\drm,y)}{ny^{n-1}}{r_\drm^{-n}}\dv y.
\end{align*}
\end{theorem}

\subsection{The PMF of $\summ$ under reduced Palm}
	Using the approach similar to Section \ref{sec:CD:PMF}, we get:
\begin{corollary}\label{cor:PMFNND}
	The PMF of $N$ is 
	\begin{align*}
	&\palmprobx{\summ=m}{!\ob}=
	\sum_{i=0}^{m}{\Delta_i(r)}q_{m-i}(r),
	\end{align*}
where 
$\Delta_{i}=\cF_{R_{i+1}}(r)-\1(i\geq1)\cF_{R_{i}}(r)$ and
$$q_{j}(r)=\frac1{j!}
	\int_{0}^{r_\drm}
	(\lambda_\drm \A(r,r_\drm,y))^{j}e^{-\lambda_\drm \A(r,r_\drm,y)}
	\frac{ny^{n-1}}{r_\drm^{n}}\dv y.$$
\end{corollary}
\begin{proof}
See Appendix \ref{CDF_nnd_proof}.
\end{proof}

Now, using Corollary \ref{cor:PMFNND} in \eqref{prob_cdf_nnd} and some manipulations (See Appendix \ref{app:4} for details), we get the following theorem. 
\begin{theorem}
	The CDF of the $k$th NND $R^{'}_{k}$ for the MCP is:
\begin{align}
F_{R^{'}_k}(r)&=1-\sum_{i=1}^{k}q_{k-i}(r)\cF_{R_{i}}(r)\label{version3}
\end{align}
where 
$ \overline{F}_{R_k}(r)$ denotes the complementary CDF (CCDF) of $ R_{k}$ i.e.
	$\overline{F}_{R_k}(r)=1-F_{R_k}(r).$
\end{theorem}

\begin{corollary}\label{nnd_thm_2}
	The CDF of $R^{'}_{k}$ for $k=1,2,3$ is:
\begin{align*}
&F_{R^{'}_1}(r)=1-\cF_{R_1}(r)q_{0}(r),\\
&F_{R^{'}_2}(r)=1-q_{1}(r)\cF_{R_1}(r)-q_{0}(r)\cF_{R_2}(r),\\
&F_{R^{'}_3}(r)=1-q_{2}(r)\cF_{R_1}(r)-{q_{1}(r)}\cF_{R_2}(r)-{q_{0}(r)}\cF_{R_3}(r).
\end{align*}
\end{corollary}
Fig. \ref{Fig:1}  shows CDF for $k$th CD and NND for  MCP obtained from analysis and simulation which validates our analysis.
\begin{figure}[t]
	\centering
	\begin{minipage}[t]{.5\linewidth}
		\includegraphics[width=1\textwidth]{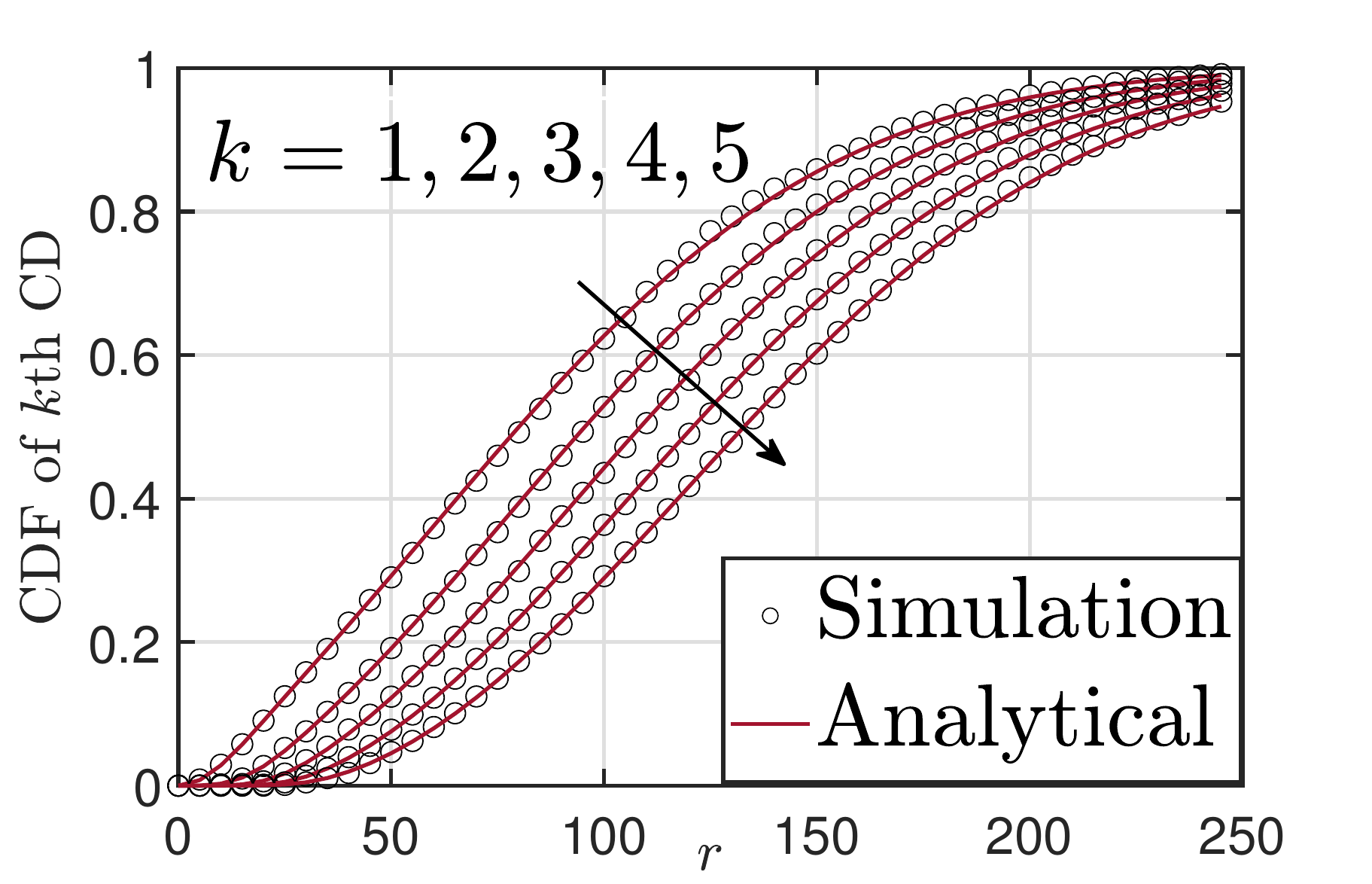}
	\end{minipage}%
	\hfill%
	\begin{minipage}[t]{.5\linewidth}
		\includegraphics[width=1\textwidth]{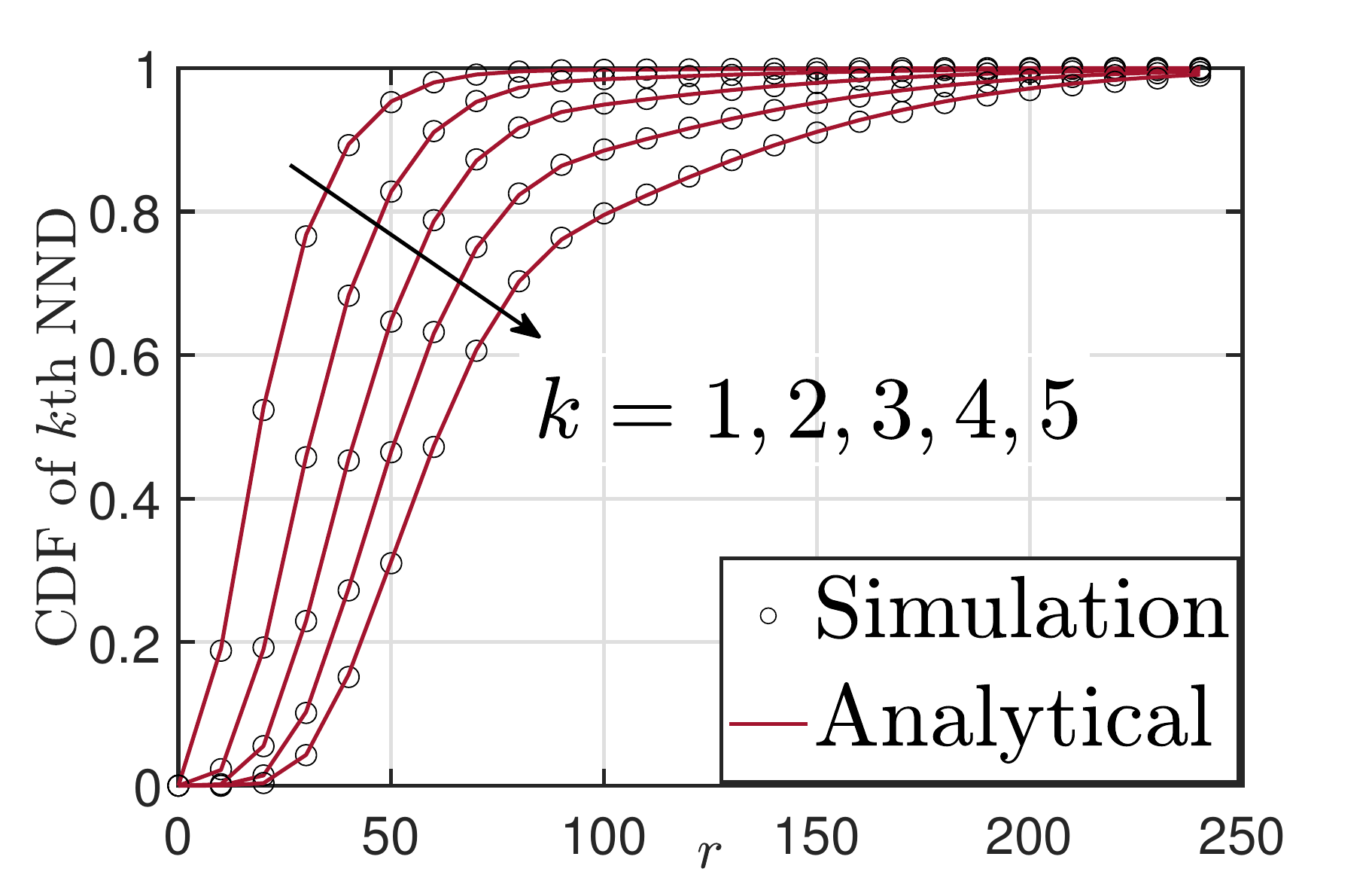}
	\end{minipage}
	\caption{The CDFs of $k$th CD and NND  with $\lambda_{\pt}=2\times10^{-5},\mo=5,r_\drm=50$. }\label{Fig:1} \vspace{-.2in}
\end{figure}
\section{Relation between CD and NND Distributions}
 \eqref{version3} establishes a relation between the CDF of $R_{k}$ and $R^{'}_{k}$. Note that \begin{align}\label{convex_combo}
		&\lim_{k\rightarrow\infty}\sum\nolimits_{i=0}^{k}q_{j}(r)= 1,
		\end{align}
and therefore $q_j\le 1$ $\forall$ $j$. This means that the CCDF of $R^{'}_{k}$ is a linear combination of CCDFs of $\{R_{i},\,\forall i\leq k\}$. Moreover, $q_j(r)$ denotes the probability that the typical point has more intra-cluster $j$ points within distance $r$ from it.
	 \subsection{The first order stochastic dominance of $R_k$ over $R^{'}_k$:} The stochastic dominance of $R_1$ over $R^{'}_1$ is  proved in  \cite{afshang2017nearest}. 
	 Note that 
	 \begin{align*}
	 \cF_{R^{'}_k}\!(r)&=\sum_{i=1}^{k}q_{k-i}(r)\cF_{R_{i}}\!(r)\le \sum_{i=1}^{k}q_{k-i}\!(r)\cF_{R_{k}}\!(r)\le \cF_{R_{k}}(r).
	 \end{align*}
	 This implies the stochastic dominance of $R_{k}$ over $R^{'}_k$ $\forall$ $k$. 	 
	  \subsection{Asymptotic results with respect to $r_\drm$ and fixed $\mo$:} As we decrease the cluster radius $r_\drm\rightarrow0$, $\A(r,r_\drm,y)\rightarrow v_n r_\drm^n$. Therefore, the CDF of NND simplifies to
	\begin{align}
	&F_{R^{'}_k}(r)=1-e^{-\mo}\sum_{i=1}^{k}\frac{{\mo}^{k-i}\cF_{R_{i}}(r)}{(k-i)!}.
	\end{align}
	Similarly, as $r_\drm\rightarrow\infty$,  $\A(r,r_\drm,y)\rightarrow v_n r^n$. Hence, $\mathcal{P}_\summ(s)\rightarrow \expS{\pi\lambda_\pt \mo (s-1_r^n}$ which is the PGF of $\summ$ for  a PPP with density $\lambda_\pt \mo$. This implies that $k$ CD distribution of MCP also converges to that of PPP. Also, from  Theorem \ref{thm:3}, 
		\begin{align*}
	&{\mathcal{P}}^{!\ob}_{\num}(s)=\mathcal{P}_{N}(s)\implies
	\lim_{r_\drm\rightarrow\infty}F_{R^{'}_{k}}(r)\rightarrow F_{R_k}(r).
	\end{align*}
\section{Applications in Large Wireless Networks}
We now discuss two applications to understand how obtained results can be used in the analysis of  wireless networks.
	
\newcommand{\kconnprob}{\mathrm{p}_k}
\subsection{Cellular networks with macro-diversity}
To enhance the communication reliability, cellular networks can utilize macro-diversity \cite{GupAndHea2018} where a user can receive transmission from multiple BSs. Let us consider a cellular network with BSs deployed as a MCP and a typical user at the origin. We assume that the user can establish a connection to a BS if its distance is less than a certain range $R$, termed connectivity radius.  Then, $k$-connectivity probability ($\kconnprob$) is defined as the probability that the user can connect to at least $k$ BS \ie
\begin{align*}
\kconnprob=\mathbb{P}[R_{k}\leq R]=F_{R_k}(R).
\end{align*}
We are interested in studying the impact of cluster radius on the connectivity probability.  Fig. \ref{Fig:2} shows the variation of $\kconnprob$ with $r_\drm$, while keeping $\mo$ constant, to show the impact of clustering. The left extreme $r_\drm\rightarrow0$ denotes very high level of clustering and right extreme denotes no clustering (PPP behavior).
The main observations are as follows:

\textbf{$k=1$:}  
$\mathrm{p}_1$ always increases with $r_\drm$ implying that clustering hurts the first connection. This can be shown from the expression of $F_{R_1}(R)$. This behavior is observed because higher $r_\drm$ causes nodes to spread more, making the closest node come within $R$ with higher probability. 

 \textbf{$k>1$:}  An increase in $r_\drm$ spreads nodes. With $r_\drm$, the following two competing factors determine $\kconnprob$'s behavior
 \begin{itemize}
 \item[-] if the $k$th closest node belongs to the same  cluster of the first closest node,  it moves away from $\ob$ decreasing $\kconnprob$
 \item[-] if the $k$th closest node belongs to any other cluster, this node comes closer  increasing $\kconnprob$.
\end{itemize}
If the  $\lambda_\pt$ is low, then first case is more likely, while second case is more likely for high $\lambda_\pt$. Therefore, $\lambda_\pt$ determines which one of the two factor dominates. We can observe that, with $r_\drm$, $\kconnprob$   increases for high $\lambda_\pt$ and decreases for low $\lambda_\pt$.

\begin{figure}[ht]
	\centering
	\includegraphics[width=.8\textwidth]{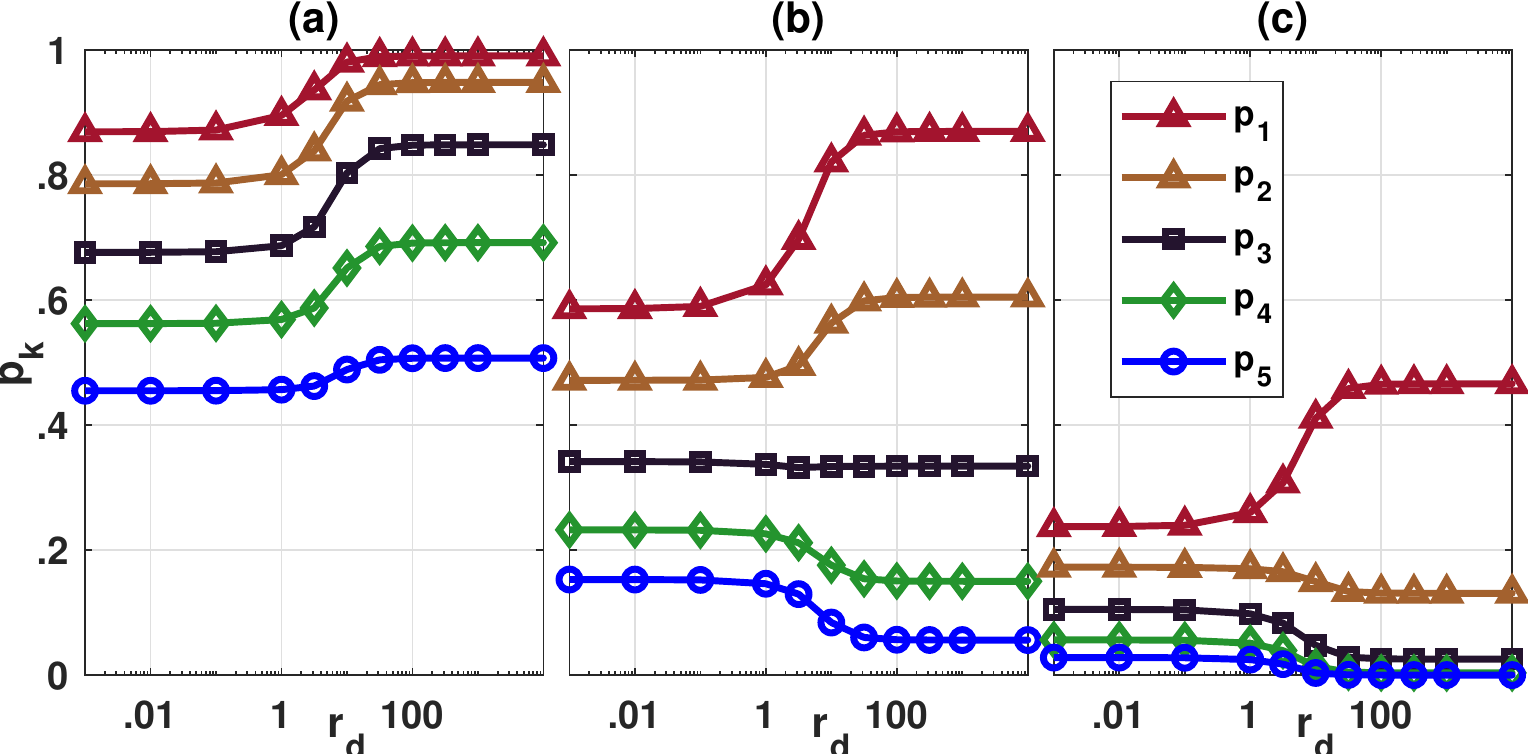}
	\caption{The variation of $k$-connectivity probability with $r_\drm$ for three values of $\lambda_{\pt}$: 
	(a) $3\times10^{-2}$, (b) $1.3\times10^{-2}$ and (c) $0.4\times10^{-2}$. Here $\mo=2$, $R=5$.} \label{Fig:2}
	\vspace{-.2in} 
\end{figure}
\newcommand{\frp}{\mathrm{f}_k}
\begin{figure}[ht]
	\centering
	\includegraphics[width=.8\textwidth]{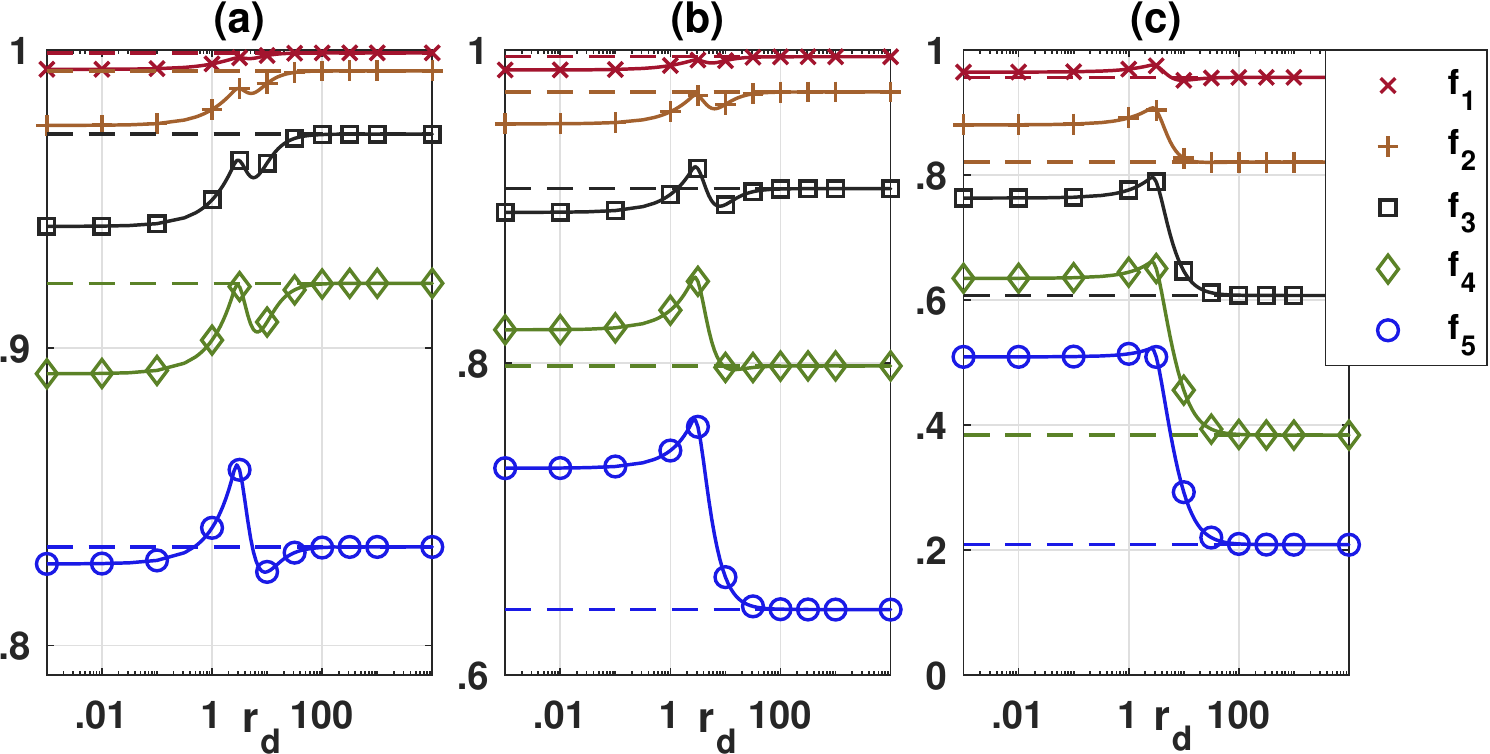}
	\caption{The variation of $\frp$ with $r_\drm$ for three values of $\lambda_{\pt}$: 
	(a) $4.5\times10^{-2}$, (b) $3.5\times10^{-2}$ and (c) $2\times10^{-2}$. Here $\mo=2$, $R=5$.} \label{Fig:3}\vspace{-.2in}  
\end{figure}

\subsection{Caching  in D2D networks}
To reduce the file access time and dependency on  central servers, nodes in a D2D network can implement caching where each node keeps copies of some popular files locally. A node can access content available at other nodes in its communication range $R$ \cite{MalShaAnd2016}. To have a right balance between space  and latency requirements, only few files are kept at each node while ensuring that all popular files are kept at at least one node of $k$ neighboring nodes. Consider a D2D networks with nodes located as MCP and a typical node $\z_\ob$ at the origin requesting a file. The   performance of the D2D network can be measured by file retrieval probability or cache hit probability which  is defined as the the probability that at least $k$ neighbors are in the communication range of a typical node, \ie
\begin{align*}
\frp=\mathbb{P}^{\ob}[R'_{k}\leq R]=F_{R'_k}(R).
\end{align*}
Our goal is to understand the impact of $r_\drm$ on $\frp$.  Fig. \ref{Fig:3} shows the variation of $\frp$ with $r_\drm$, while keeping $\mo$ constant, to show the impact of clustering. The dashed lines represent the corresponding $\frp$ for a PPP with density $\mo \lambda_\pt$. 
 The observed behavior can be justified in the following way.  Here also, with $r_\drm$, there are two competing factors determining $\frp$'s behavior
 \begin{itemize}
 \item[-] if the $k$th neighbor belongs to the typical  cluster,  it moves away from $\z_\ob$ decreasing $\frp$
 \item[-] if the $k$th neighbor belongs to any other cluster, this nodes comes closer  increasing $\frp$.
\end{itemize}

When $r_\drm$ is very small, all nodes of the typical cluster are within communication range. Changing $r_\drm$ doesn't affect $\frp$. 

Now, at lower $\lambda_\pt$, increasing $r_\drm$ can spill intra-cluster nodes outside the communication range, decreasing the $\frp$ significantly. Since other clusters are far away, changing $r_\drm$ doesn't bring them in communication range. At large $r_\drm$, nodes start exhibiting the PPP behavior, making $\frp$ independent of $r_\drm$.

At higher $\lambda_\pt$, nodes from other clusters start playing role in determining $\frp$. As $r_\drm$ increases, these nodes can come inside the connection region and become one of $k$ neighbors which increases $\frp$. As $r_\drm$ further increases, intra-cluster nodes goes outside the connection range decreasing $\frp$ sharply. This fall continues until $r_\drm$ becomes significantly larger than $R$. After that, intra-cluster nodes are outside the connection radius with high probability. Hence, $\frp$ again starts increasing due to increasing proximity of nodes of other clusters. At large $r_\drm$, $\frp$ becomes independent of $r_\drm$ due to limiting PPP nature.

\appendices
\section{} \label{Appn:1}

The  PGF of $\summ$ is defined as
\begin{align*}
&\mathcal{P}_{\summ}(s)=\mathbb{E}\left[s^{\summ}\right]=\expect{\prod_{\Z\in\Phi} s^{\indside{\Z \in \Ball(\ob,r)}}}
\end{align*}
which is nothing but the PGFL $\PGFL[v]$ of MCP for $v(\x)=s^{\1(\x\in\B(\ob,r))}$.
	Further note that 
\begin{small}
	\begin{align*}
	v(\x)&=s \times \1(\x\in\Ball(\ob,r))+1 \times (1-\1(\x\in\Ball(\ob,r)))\\
	&=1+(s-1)\1(\x\in\Ball(\ob,r)).
	\end{align*}
\end{small}
From the PGFL \eqref{PGFL}, we get
\begin{small}
	\begin{align}
	\mathcal{P}_{\summ}(s)=\PGFL[v]&=\exp\left(-\lambda_{\pt}\int_{\R^{n}}\left(1-\mathcal{H}_{\x}[v] \right)\dv \x\right)\label{app1:pgfl}
	\end{align}
\end{small}
with
\begin{small}
		\begin{align}
	\mathcal{H}_{\x}[v] &=\exp\left(-\lambda_\drm
	\int_{\Ball(\ob,r_\drm)}\left(1-v(\x+\y)\right)\dv\y
	\right)\nonumber\\
	& =\exp\left(\lambda_\drm
	\int_{\Ball(\ob,r_\drm)} 
	\left((s-1)\1(\x+\y\in\Ball(\ob,r))\right)\dv\y
	\right)\nonumber\\
	& =\exp\left(\lambda_\drm
	(s-1)\A(r,r_\drm,\dist{\x})
	\right).\label{eq:Hxv}
	\end{align}	
\end{small}
Replacing the expression of $\mathcal{H}_{\x}[v]$ in  \eqref{app1:pgfl} we get
\begin{small}
		\begin{align*}
	&&\mathcal{P}_{\summ}(s)=\exp\left(-\lambda_{\pt}\int_{\R^{n}}\left(1-e^{\lambda_{\drm}\A(r,r_\drm,\dist{\x})(s-1)}\right)\dv \x \right)\label{derivative}.
	\end{align*}
\end{small}
Simplifying the above expression, we get the Theorem \ref{Thm:1}.
\section{}\label{appen_nnd_1}
The typical point $\z_\ob$ is located at the origin $\ob$. Under reduced Palm, the number of points falling in the ball $\Ball(\ob,r)$ is:
\begin{small}
	\begin{align*}
	\num&=\sum\nolimits_{\Z_i\in\Phi\setminus\{\z_\ob\}}\1(\Z_i\in\Ball(\ob,r)).
	\end{align*}
\end{small}
The  PGF of $N$ under reduced Palm is 
\begin{small}
	\begin{align*}
	&\mathcal{P}^{!\ob}_{N}(s)=\mathbb{E}^\ob\left[s^{\sum_{\Z_i\in\Phi\setminus\{\ob\}}\1(\Z_i \in \Ball(\ob,r))}\right],\\
	&=\mathbb{E}^\ob\left[\prod_{\Z_i\in\Phi\setminus\{\ob\}}s^{\1(\rx_i\in \Ball(\ob,r))}\right]=\mathbb{E}^{!\ob}\left[\prod_{\Z_i\in\Phi}s^{\1(\Z_i\in \Ball(\ob,r))}\right],
	\end{align*}
\end{small}
which is nothing but the PGFL under reduced Palm $\PGFL^{!}_{\ob}[v]$ of MCP (given in \eqref{eq:condpgfl}) for
\begin{center}
$v(\x)=s^{\1(\x\in\B)}=1+(s-1)\1(\x\in\B)$.
\end{center}
Therefore
\begin{small}
	\begin{align*}
	\mathcal{P}^{!\ob}_{N}(s)
	&=\mathcal{P}_N(s)\ \frac{1}{v_n r_\drm^{n}}\int_{\B}\!\!\!\!\! \mathcal{H}_{-\x_\ob}[v]\dv \x_\ob
	\end{align*}
\end{small}
Now using \eqref{eq:Hxv} and converting to polar coordinates, we get
\begin{small}
		\begin{align*}
	&\mathcal{P}^{!\ob}_{\num}(s)=\mathcal{P}_{\summ}(s)\int_{0}^{r_\drm}e^{(s-1)\lambda_\drm\A(r,r_\drm,x)}\frac{nx^{n-1}}{r_\drm^{n}}\dv x.
	\end{align*}
\end{small}
\vspace{-2.1em}
\section{}\label{CDF_nnd_proof}
 Applying PGF-PMF relationship similar to Section \ref{sec:CD:PMF}-B and Theorem \ref{thm:3}, the PMF of $N$ under Palm is given as
\begin{small}
	\begin{align*}
	&\mathbb{P}^{!\ob}[N=m]=\frac1{m!}\frac{\dv^m}{\dv s^m}\mathcal{P}^{!\ob}_{N}(s)\bigconditioned_{s=0}\\
	&\stackrel{(a)}=\frac{1}{m!}\sum_{i=0}^{m}{m\choose i}\mathcal{P}^{(i)}_{N}(0)\int_{0}^{r_\drm} {(\lambda_\drm\A(r,r_\drm,x))}^{m-i}e^{-\lambda_\drm\A(r,r_\drm,x)}\frac{nx^{n-1}}{r_\drm^{n}}\dv x,
	\end{align*}
\end{small}
where $(a)$ is achieved using general Leibniz rule. Using \eqref{PGF_PMF_relationship} and definition of $q_j()$, 	
\begin{small}
		\begin{align*}
	\mathbb{P}^{!\ob}[N=m]
	&=\sum_{i=0}^{m}\frac{1}{m!}{m\choose i}i!\mathbb{P}[N=i]q_{m-i}(r) (m-i)!\\
	&=\sum_{i=0}^{m}\mathbb{P}[N=i]q_{m-i}(r).
	\end{align*}
\end{small}
Now noting that $\mathbb{P}[N=i]=\cF_{R_{i+1}}(r)-\1(i\geq1)\cF_{R_{i}}(r)$, we get the desired result.

\section{}\label{app:4}
Now, using Corollary \ref{cor:PMFNND} in \eqref{prob_cdf_nnd},  the CDF of $R^{'}_{k}$ is
\begin{small}
	\begin{align*}
	&F_{R^{'}_k}(r)=1-\sum_{m=0}^{k-1}\left[\sum_{i=0}^{m}q_{m-i}(r)\cF_{R_{i+1}}(r)-\sum_{i=1}^{m}\cF_{R_{i}}(r)q_{m-i}(r)\right].
	\end{align*}
	\end{small}
Replacing $i^{'}=i+1$, and $m^{'}=m+1$  in the first term, enables us to cancel all terms except the term corresponding to  $m=k-1$ which gives the desired result. 

\vspace{12pt}
\color{red}
\end{document}